\documentclass[conference]{sig-alternate-10pt}

\usepackage{algorithmic}
\usepackage{algorithm}
\usepackage{url}
\usepackage{verbatim}
\usepackage{enumerate}
\usepackage{fancybox}
\usepackage{eufrak}
\usepackage{amsmath, amssymb, framed}
\usepackage{verbatim}
\usepackage{todonotes}
\usepackage{multirow}
\usepackage{array}
\usepackage{graphicx}
\usepackage{balance}
\usepackage{todonotes}
\usepackage{bytefield}
\usepackage{paralist}
\usepackage{float}
\usepackage[footnotesize]{subfigure}
\usepackage{tikz}
\usetikzlibrary{trees,chains,shapes.multipart}
\usepackage{pgf-umlsd}
\usepgflibrary{arrows} 
\usepackage{enumitem}
\usepackage{hyphenat}

\setenumerate{noitemsep,topsep=0pt,parsep=0pt,partopsep=0pt}
\setitemize{noitemsep,topsep=0pt,parsep=0pt,partopsep=0pt}

{\begin{Sbox}\begin{minipage}}%
{\end{minipage}\end{Sbox}\fbox{\TheSbox}}

\newtheorem{theorem}{Theorem}
\newdef{definition}{Definition}

\newcommand{\ccnname}[1]{\path{#1}}
\newcommand{\Adv}{Adv}
\newcommand{\mAdv}{\mathrm{Adv}}
\newcommand{\auth}{\mathbb{U}}
\newcommand{\nauth}{\bar{\mathbb{U}}}

\begin{document}

\title{Interest-Based Access Control for Content Centric Networks (extended version)}

\numberofauthors{1}
\author{
\alignauthor
Cesar Ghali \quad Marc A. Schlosberg \quad Gene Tsudik \quad Christopher A. Wood\\
       \affaddr{University of California Irvine}\\
       \email{\{cghali, marc.schlosberg, gene.tsudik, woodc1\}@uci.edu}
}

\maketitle

\maketitle

\begin{abstract}
Content-Centric Networking (CCN) is an emerging network architecture 
designed to overcome limitations of the current IP-based Internet. 
One of the fundamental tenets of CCN is that data, or content, is 
a named and addressable entity in the network. Consumers request content
by issuing interest messages with the desired content name. These interests
are forwarded by routers to producers, and the resulting content object is returned
and optionally cached at each router along the path. In-network caching
makes it difficult to enforce access control policies on sensitive content
outside of the producer since routers only use interest information for forwarding
decisions. To that end, we propose an Interest-Based Access Control (IBAC) scheme 
that enables access control enforcement using only information contained in
interest messages, i.e., by making sensitive content names unpredictable
to unauthorized parties. Our IBAC scheme supports both hash- and encryption-based 
name obfuscation. We address the problem of interest replay
attacks by formulating a mutual trust framework between producers and consumers
that enables routers to perform authorization checks when satisfying interests
from their cache. We assess the computational, storage, and bandwidth overhead
of each IBAC variant. Our design is flexible
and allows producers to arbitrarily specify and enforce any type of access control
on content, without having to deal with the problems of content encryption and 
key distribution. This is the first comprehensive design for CCN access control
using only information contained in interest messages.
\end{abstract}

\section{Introduction} \label{sec:intro}
The purpose of the original Internet in the 1970s was to provide end-to-end communication
for a few thousand users to access scarce and expensive resources via terminals. Since then the
number of Internet users has grown exponentially, reaching more than 3 billion, with each using a 
wide variety of applications from the dynamic web to content distribution. This shift of usage
exposed certain limitations of the IP-based Internet design and motivated exploration of new architectures.

Content-Centric Networking (CCN) is an approach to inter-networking exemplified by two well-known
research efforts: CCNx~\cite{moskoccn} and Named-Data Networking (NDN)~\cite{jacobson2009networking}. 
The main goal of CCN is to develop the next-generation
Internet architecture with an emphasis on efficient content distribution, security, and privacy.
Unlike current IP-based networking where data is requested by addressing the machine where it
is hosted, each CCN {\em content} is assigned a unique name. Users (referred to as consumers) request
content objects by issuing an {\em interest} for a given name. This interest can then be satisfied or
served from any entity (i.e., producer or router) as long as the replied content's name matches that of the interest.

To facilitate efficient content distribution, a CCN router maintains a cache.
This enables routers to satisfy interests, which reduces end-to-end latency and decreases
bandwidth utilization when requesting popular content. Since interest messages
may be satisfied by any cached version of the content, interest messages may not,
and need not, reach the producer. Therefore, enforcing content access control
\emph{within the network} is a challenge. Furthermore, even if all interests are
forwarded to producers, the latter might not be able to enforce access control since interest
messages, by design, do not carry any form of consumer identification or authentication information.

In this paper, we propose an access control scheme based on interests
-- Interest-Based Access Control (IBAC). The intuition is that if consumers are not
allowed access to certain content, they should not be able to generate the corresponding
interests, i.e., they should not be able to learn the content's name. IBAC may also be used with content encryption
to conceal both the name and the payload of the content object. However, using IBAC in isolation is
advantageous in scenarios where content object payloads may need to be modified by an
intermediate service, e.g., a media encoding application or proxy. In this case, content
encryption prevents such modifications by services or applications besides the producer.
Moreover, although IBAC involves the network layer,
we believe that this is necessary to allow routers (with caches) to enforce access control.
To be more specific, we claim that {\em any entity which serves content
should also be able to authorize interests for said content}. 

The main contributions of this paper are:

\begin{itemize}
\item Architectural modifications to support IBAC without diminishing caching benefits.
\item A mutual trust scheme wherein routers can verify whether consumers are authorized to access cached content.
\item A security analysis of the proposed IBAC scheme.
\item Evaluation of router performance overhead when serving content via IBAC compared to publicly accessible content.
\end{itemize}

The rest of this paper is organized as follows. Section \ref{sec:overview} presents an overview of CCN architectures in the
context of CCNx. Section \ref{sec:ac} then provides an overview of access control techniques for CCN, including
encryption- and interest-based access control. Section \ref{sec:secmodel} presents the relevant security definitions and
adversarial models used to assess the IBAC scheme discussed at length in Section \ref{sec:ibac}. We discuss the IBAC
security consideration in Section \ref{sec:security}. We then analyze the overhead incurred by IBAC in Section
\ref{sec:analysis} and conclude in Section \ref{sec:conclusion}.

\section{CCN Overview} \label{sec:overview}
Content Centric Networking (CCN) is one of the main Information-Centric Networking (ICN) architectures. 
Related architectures, such as Named Data Networking (NDN) \cite{NDN},
are similar, albeit with some small protocol and packet format differences. 
This section overviews ICNs in the context of the CCN protocol and CCNx reference implementation. 
Given familiarity with either CCN or NDN, it can be skipped without loss of continuity.

In contrast to TCP/IP, which focuses on end-points of communication and their names and addresses,
ICN architectures such as CCN \cite{jacobson2009networking,moskoccn} focus on content by making it
named, addressable, and routable within the network. A content name is a URI-like \cite{berners2005rfc} name
composed of one or more variable-length name components, each separated by a \ccnname{/}
character. To obtain content, a user (consumer) issues a request, called an \emph{interest} message,
with the name of the desired content. This interest will be \emph{satisfied} by either
(1) a router cache or (2) the content producer. A \emph{content object} message is returned to
the consumer upon satisfaction of the interest. Moreover, name matching in CCN is
exact, e.g., an interest for \ccnname{lci:/facebook/Alice/profile.html} can only be satisfied by
a content object named \ccnname{lci:/facebook/Alice/profile.html}.\footnote{Name matching is not exact in NDN \cite{NDN}.}

Aside from the content name, CCN interest messages may include the following fields:
\begin{itemize}
\item \texttt{Payload} -- enables consumers to push data to producers along with the
request.\footnote{Currently, NDN interest messages do not provide an arbitrary-length payload and
therefore cannot support the proposed IBAC scheme. However, if in the future the NDN
interest format is modified to include a field similar to the CCNx payload, our IBAC scheme
will become applicable.}
\item \texttt{KeyID} -- an optional hash digest of the public key used to verify
the desired content's digital signature. If this field exists, the network guarantees
that only content objects which can be verified with the specified key will be returned
in response to an interest.
\item \texttt{ContentObjectHash} -- an optional hash value of the content being requested. If this
field exists, the network guarantees the delivery of the exact content that consumer requests.
\end{itemize}
CCN content objects include several fields. In this work, we are only interested in the following three:
\begin{itemize}
\item \texttt{Name} -- a URI-like name formatted as a sequence of \ccnname{/}-separated name components.
\item \texttt{Validation} -- a composite of validation algorithm information (e.g., the signature algorithm used,
its parameters, and a link to the public verification key), and validation payload (e.g., the signature). We use the
term ``signature'' to refer to this field.
\item \texttt{ExpiryTime} -- an optional, producer-recommended time for the content objects to be cached.
\end{itemize}
There are three basic types of entities in CCN that are responsible for transferring interest and 
content object messages:\footnote{A physical entity, or host, can be both a consumer and producer of content.}
\begin{itemize}
\item {\em Consumer} -- an entity that issues an interest for content.
\item {\em Producer} -- an entity that produces and publishes content.
\item {\em Router} -- an entity that routes interest packets and forwards corresponding content packets.
\end{itemize}
Each CCN entity must maintain the following two components:
\begin{itemize}
\item {\em Forwarding Interest Base} (FIB) -- a table of name prefixes and corresponding outgoing interfaces. The FIB is used to route interests based on longest-prefix-matches of their names.
\item {\em Pending Interest Table} (PIT) -- a table of outstanding (pending) interests and a set of corresponding incoming interfaces.
\end{itemize}
An entity may store an optional {\em Content Store} (CS), which is a buffer used for content caching
and retrieval. Again, the timeout of cached content is specified in the \texttt{ExpiryTime} field of the
content header. From here on, we use the terms {\em CS} and {\em cache} interchangeably.

Router entities use the FIB to forward interests from consumers to producers, and then later
use the PIT to forward content object messages along the reverse path to the consumer.
More specifically, upon receiving an interest, a router $R$ 
first checks its cache to see if it can satisfy this interest locally.
Producer-originated digital signatures allow consumers to authenticate received content, regardless of the
entity that actually served the content. Moreover, the Interest-Key Binding rule (IKB)~\cite{ghali2014elements}
enables routers to efficiently verify received content signatures before caching, in order to avoid
content poisoning attacks~\cite{ghali2014needle}. Essentially, consumers and producers provide routers with
the required trust context to enable efficient signature verification.

When a router $R$ receives an interest for name $N$ that is not cached and there are no
pending interests for the same name in its PIT, $R$ forwards the interest to the next
hop according to its FIB. For each forwarded interest, $R$ stores some amount of state information,
including the name of the interest and the interface from which it arrived, so that
content may be sent back to the consumer. If an interest for $N$ arrives while there is
already an entry for the same content name in the PIT, $R$ only needs to update the
arriving interface. When content is returned, $R$ forwards it to all of the corresponding
incoming interfaces, and the PIT entry is removed. If a router receives a content object
without a matching PIT entry, the message is deemed unsolicited and subsequently discarded.

\section{Access Control Overview} \label{sec:ac}
One key feature of CCN is that content is decoupled from its
source; there is no notion of a secure channel between a consumer and producer.
Consequently, ensuring that only authorized entities have access to content is a
fundamental problem. In this section, we explore complementary 
approaches to access control: (1) content encryption and (2) interest name
obfuscation and authorization. 

\subsection{Encryption-Based Access Control}
The most intuitive solution to the access control problem is via encrypted content
which can only be decrypted by authorized consumers possessing the appropriate
decryption key(s). This enables content objects to be disseminated
throughout the network since they cannot be decrypted by adversaries 
without the appropriate decryption key(s).

Many variations of this approach have been proposed \cite{Smetters2010,Misra2013,Ion2013,Wood2014}.
Kurihara et al. \cite{ifip15} generalized these specialized approaches in
a framework called CCN-AC, an encryption-based access control framework
to implement, specify, and enforce access policies. It uses CCN manifests\footnote{Manifests are special 
types of content that are used to provide structure and additional information to otherwise flat 
and simple content objects \cite{moskoccn}.}
to encode access control specification information for a particular set of
content objects. Consumers use information in the manifest to (1) request appropriate
decryption keys and (2) use them to decrypt the content object(s) in question.

Outside of ICN, there have been many proposed access control frameworks based on encryption.
Recently, access control in shared cloud storage or social network
services, e.g., Google Drive, Dropbox, and Facebook, generated a great 
deal of attention from the research community~\cite{Zhou2013,Wang2010a,Yu2010,Jahid2011}.
For instance, Kamara et al. \cite{Kamara2010} modeled encryption-based access control framework
for cloud storage. Microsoft PlayReady \cite{MSPlayReady} is another popular access control
framework for encrypted content dissemination over the Internet.

Despite its widespread use, encryption-based access control causes potentially prohibitive overhead 
for both producers and consumers. It most cases where hybrid encryption is used, it also requires 
keys to be distributed alongside each content object, which introduces another consumer-to-producer
message exchange. Also, encryption-based access control does not provide flexibility if
content objects need to be modified by an intermediate service, e.g., a media
encoding or enhancement application. Content encryption prevents such post-publication
modifications without revealing the secret decryption key(s) to such services.

\subsection{Interest-Based Access Control}
Interest-based access control (IBAC) is an alternative technique, though not mutually exclusive
with content encryption, for implementing access control in CCN. It is based on interest {\em name obfuscation}
and {\em authorized disclosure}. Name obfuscation hides the \emph{target} of an
interest from eavesdroppers. As mentioned in \cite{jacobson2009networking},
name obfuscation has no impact on the network since routers use only the binary
representation of a name when indexing into PIT, CS, and FIB. As long as producers generate
content objects with matching names, the network can seamlessly route interests and
content objects with obfuscated names. However, interests with
obfuscated names must contain \emph{routable prefixes} so that they can be
forwarded from consumers to the producers. In other words, only a 
subset of name components (e.g., suffix of the name) is obfuscated.

Another goal of name obfuscation is to \emph{prevent unauthorized users from creating 
interests for protected content}. In other words, if a particular consumer $Cr$ is not permitted
to access content with name $N$, $Cr$ should not be able to generate $N' = f(N)$, where
$f(\cdot)$ is some obfuscation function that maps $N$ to an obfuscated name $N'$. 
For routing purposes, only the \emph{suffix} of the name is obfuscated; there must exist
a cleartext prefix that is used to route the interest with a partially obfuscated name
to the intended producer. Possible obfuscation functions include keyed cryptographic 
hash functions and encryption algorithms. We explore both possibilities in this paper.

Authorized disclosure is the second element of IBAC. This property implies that
any entity serving content must authorize any interest for said content before it is served. 
In this context, authorization is necessarily coupled with authentication 
so that the entity serving the content can determine the identify of the requesting consumer. 
Therefore, consumers must provide sufficient authentication information, e.g., via an interest signature.
Thus, to implement authorized disclosure (in the presence of router caches), any entity serving content
must (a) possess the information necessary to perform authentication and
authorization checks and (b) actually verify the provided authentication
information. This issue is discussed at length in Section \ref{sec:acky}.
It is worth mentioning that disabling content caching defers authorized disclosure checks to producers. 
In this case, all interests will be forwarded to producers that posses the 
information needed to perform these checks. However, by itself, prohibiting content from being cached 
is \emph{not} a form of access control and reduces the effectiveness of content retrieval.

Fotiou et. al. \cite{fotiou2012access} proposed an access control mechanism similar to IBAC
for non-ICN architectures, and conjectured that it should be applicable to ICNs.
In \cite{fotiou2012access}, access control computation and overhead are delegated to
a separate, non-cache entity. This entity, known as the access control provider, maintains
access control policies given by a specific producer. Each content object has a pointer to
a function that determines whether or not to serve the content to the requesting consumer,
and the access control provider is responsible for evaluating this function. Content
objects are stored at relaying parties, which are oblivious to the specific access control
policy protecting the content objects. Similarly, the access control provider has no
knowledge of the consumer requesting the content (for user privacy purposes), and just
evaluates whether the relaying party should forward the content object. The cache, in
this scenario, is not responsible for the extra computational overhead. This approach
is different from our work in that we (1) maintain the association between content and
authorization, and (2) provide routers with an efficient authorization verification method,
thus eliminating the need for an external access control provider.

\section{Security Model} \label{sec:secmodel}
Let $\auth(N)$ denote the set of authorized consumers for a content
object with name $N$ generated and controlled by a producer $P$, and let
$\nauth(N)$ be its complement, i.e., the set of all unauthorized consumers.
Let $\mathsf{Path}(Cr, P)$ be the set of all routers on the path between the
consumer $Cr \in \auth(N)$ and $P$. We assume the existence of
an adversary \Adv\ who can deploy and compromise
unauthorized consumer any any router $R \notin \mathcal{R}$.\footnote{Any one of these actions can be performed
adaptively, i.e., in response to status updates or based on observations.} To keep
this model realistic, we assume that the time to mount such an attack is
non-negligible, i.e., longer than the average RTT for a single interest-content
exchange. Table \ref{tab:acronym} summarizes the notation used in the rest of this paper.

Formally, we define \Adv\ as a 3-tuple:
$
(\mathcal{P}_{\mAdv} \setminus \{P\}, \mathcal{C}_{\mAdv} \setminus \auth(N), \mathcal{R}_{\mAdv} \setminus \mathsf{Path}(Cr, P))
$
where the components denote the set of compromised producers, consumers, and routers,
respectively. If \Adv\ controls a producer or a consumer then it is assumed to have
complete and adaptive control over how they behave in an application session.
Moreover, \Adv\ can control all of the timing, format, and actual information of
each content through compromised nodes and links.

Let $\mathsf{Guess}$ denote the event where \Adv\ correctly recovers the
obfuscated form of a content name. Let $\mathsf{Bypass}$ denote the event
where \Adv\ successfully bypasses the authorization check for a protected
content object. We define the security of an IBAC scheme with respect to
these two events as follows.

\begin{definition}
\label{def:secure_nbac_1}
An IBAC scheme is {\em secure, but subject to replay
attacks,} if $\Pr[\mathsf{Guess}] \leq \epsilon(\kappa)$ for any negligible
function $\epsilon$ and a security parameter $\kappa$.
\end{definition}

\begin{definition}
\label{def:secure_nbac_2}
An IBAC scheme is {\em secure in the presence of replay
attacks,} if $\Pr[\mathsf{Guess} + \mathsf{Bypass}] \leq \epsilon(\kappa)$ for
any negligible function $\epsilon$ and a security parameter $\kappa$.
\end{definition}

Replay attacks are artifacts of the environment where CCN 
access control scheme is deployed. In other words, in
networks where links are insecure, passive eavesdroppers can observe previously
issued interests and replay them for protected content. Consequently, these attacks
are considered orthogonal to the security of the underlying obfuscation scheme used for
access control enforcement. The authorized disclosure element of IBAC is intended
to prevent such replay attacks.

To justify our adversarial limitation to off-path routers, consider the following scenario.
If \Adv\ can compromise a router $R \in \mathsf{Path}(Cr, P)$,
then \Adv\ is able to observe \emph{all} content that flows along this path.
Therefore, we claim that on-path adversaries motivate access control schemes 
based on content encryption; IBAC will not suffice. Moreover, we exclude adversaries 
capable of capturing interests and replaying them in other parts of the network -- see 
Section \ref{sec:replay} for details.

\begin{table}[t]
\begin{center}
\caption{Relevant notation.}
\label{tab:acronym}
\small
\begin{tabular}{|c|l|} \hline
{\bf Notation} & {\bf Description} \\ \hline
\Adv\ & Adversary \\ \hline
$Cr$ & Consumer \\ \hline
$P$ & Producer \\ \hline
$\mathsf{prefix}$ & Producer prefix \\ \hline
$N$ & Content name in cleartext \\ \hline
$N'$ & Obfuscated content name \\ \hline
$I[N]$ & Interest with name $N$ \\ \hline
$CO$ & Content object \\ \hline
$CO[N]$ & Content object with name $N$ \\ \hline
$\mathsf{ID}(\cdot, \cdot)$ & Key identifier function \\ \hline
$f(\cdot)$ & Obfuscation function \\ \hline
$\mathsf{enc}(\cdot, \cdot), \mathsf{dec}(\cdot, \cdot)$ & Symmetric-key encryption and \\
& decryption function \\ \hline
$\mathsf{Enc}(\cdot, \cdot), \mathsf{Dec}(\cdot, \cdot)$ & Public-key encryption and \\
& decryption function \\ \hline
$\mathsf{H}(\cdot)$ & Cryptographic hash function \\ \hline
$\auth(N)$ & Set of authorized consumers \\ \hline
$\mathbb{G}_i$ & Access control group $i$ \\ \hline
$k_{\mathbb{G}_i}$ & Obfuscation key of group ${\mathbb{G}_i}$ \\ \hline
$pk^s_{\mathbb{G}_i}, sk^s_{\mathbb{G}_i}$ & Public and private signing key pair \\
& associated with group $\mathbb{G}_i$ \\ \hline
$\kappa$ & Global security parameter \\ \hline
$\mathbb{C}$ & Set of all content objects \\ \hline
$r, t$ & nonce and timestamp \\ \hline
$\mathbb{B}$ & Nonce hash table \\ \hline
\end{tabular}
\end{center}
\end{table} 
\section{IBAC by Name Obfuscation} \label{sec:ibac}
Recall that the intuition behind IBAC is that if consumers are not allowed to access
certain content, they should not be able to issue a ``correct'' interest for it.
Specifically, only a consumer $Cr \in \auth(N)$ should be able to derive the
obfuscated name $N'$ of an interest requesting content with name $N$ provided by producer $P$. In this
section, we discuss two types of name obfuscation functions:
(1) encryption functions and (2) hash functions. 


\subsection{Encryption-Based Name Obfuscation}
Let $\mathsf{Enc}(k, N)$ be a \emph{deterministic} encryption function which
takes as input a key $k \in \{0,1\}^{\kappa}$ and an arbitrary
long non-empty binary name string $N$, and generates an encrypted
name $N'$. Let $\mathsf{Dec}(k, N')$ be the respective 
decryption function. With encryption, the goal is for authorized clients to encrypt components of 
a name so that the producer can perform decryption to identify
and return the appropriate content object.\footnote{Recall that a cleartext name prefix is needed
to route the interest to the intended producer.} Obfuscation is based
on knowledge of the encryption key and the content name under IBAC protection.
In other words, even if an adversary knows the name $N$, it cannot generate $N'$
since it does not possess the appropriate key. 

To illustrate how encryption-based obfuscation would work, assume first 
that $Cr$ uses $k$ to generate $N'$ as $N' = \mathsf{Enc}(k, N)$. $P$ then 
recovers $N$ as $N = \mathsf{Dec}(k, N')$ to identify the content object in 
question and returns it with the \emph{matching name} $N'$ (not $N$). 
We prove the security of this obfuscation variant of IBAC (i.e., without
authorized disclosure) in Appendix \ref{app:proofs}. \\

\noindent
\textbf{Supporting Multiple Access Groups}: Thus far, we assumed that name encryption
(obfuscation) keys are known to all authorized consumers in $\auth(N)$. However, this
might not be the case in practice. $P$ might provide content under IBAC to several
access groups each with different privileges.\footnote{We assume that each content object
is only accessible by a single access group. However, this assumption will be relaxed
later in the paper.} Specifically, consumers in groups $\mathbb{G}_i(N) \subset \auth(N)$,
for $i = 1, 2, \dots$, might be allowed access to different resources. Therefore,
several obfuscation keys, one for each group, should be utilized. For notation simplicity, 
we refer to $\mathbb{G}_i(N)$ as $\mathbb{G}_i$. Note that in an extreme scenario,
each group would only contain a single consumer, i.e., each individual consumer has
a unique key used to access the content in question. 

To decrypt the obfuscated name $N'$, $P$ must identify the obfuscation key used to generate $N'$. 
This can be achieved if such consumers specify an identifier for the
key used in the interest. Such an identifier could
simply be the digest of the obfuscation key $\mathsf{ID}_{\mathbb{G}_i} = \mathsf{H}(k_{\mathbb{G}_i})$,
where $k_{\mathbb{G}_i}$ is $\mathbb{G}_i$'s encryption key. $\mathsf{ID}_{\mathbb{G}_i}$ can be included in
the interest {\tt Payload} field. Since matching in CCN is exact, $\mathsf{ID}_{\mathbb{G}_i}$
cannot be included in interests name.

Recall that CCN interest messages, by design, do not carry any source information,
which provides some degree of anonymity. However, including $\mathsf{ID}_{\mathbb{G}_i}$
enables interest linkability by eavesdroppers (malicious or not). In other words,
$\mathsf{ID}_{\mathbb{G}_i}$ can reveal the access group identities to which
consumers belong, but not the identities of the consumers 
themselves. If this linkability is an issue for applications, $\mathsf{H}(k_{\mathbb{G}_i})$ can be encrypted using $P$'s
public key $pk^P$ in the form $\mathsf{ID}_{\mathbb{G}_i} = \mathsf{Enc}(pk^P, \mathsf{H}(k_{\mathbb{G}_i}))$.\footnote{Since a consumer cannot be expected to know the router from which content will be served, it is not plausible for them to encrypt these IDs with the public key of a (set of) router(s).}
Note that for two identifier values of the same group, i.e., with the same $k$,
to be indistinguishable, $\mathsf{Enc}(\cdot, \cdot)$ must be secure
against chosen plaintext attacks~\cite{katz2014introduction}.


\subsection{Hash-Based Name Obfuscation} \label{sec:hash-obfuscation}
Let $\mathsf{H}(k, N)$ be a keyed cryptographic hash function. The 
obfuscated name $N'$ can be generated as $N' = \mathsf{H}(k, N)$ for some key $k \in \{0,1\}^{\kappa}$.
Since hash functions are one-way, producers must maintain a hash table that maps obfuscated
names to the original content name, i.e., $\mathsf{M} : N' = H(k, N) \to N$ for all deployed 
keys.\footnote{Producers do not have to keep hash tables for all {\em possible} keys of size $\kappa$, 
only tables of keys used by producers and issued to access groups.} The size of this
hash table is $\mathcal{O}(|\mathbb{K}| \times |\mathbb{C}|)$, where 
$\mathbb{K}$ is the set of all keys and $\mathbb{C}$ is set of all content objects 
generated or published by $P$ under IBAC protection. This approach provides the same benefits of encryption-based
name obfuscation, however, it incurs additional computation and storage overhead at the producer.
Thus, while keyed hash functions are viable for name obfuscation, deterministic encryption 
is a much better approach.


\section{Security Considerations} \label{sec:security}
In this section we discuss the security of IBAC with respect to the 
adversary model described in Section~\ref{sec:secmodel}.

\subsection{Replay Attacks} \label{sec:replay}
Regardless of the obfuscation function used, both previously described IBAC schemes
are susceptible to replay attacks. This is because both obfuscation functions are deterministic. 
Therefore, an eavesdropper $\mAdv \in \nauth(N)$ could issue an
interest with a captured $N'$ and receive the corresponding content under IBAC protection
from either the producer or a router cache. In other words, the same ``feature'' that 
makes it possible for authorized consumers to fetch IBAC-protected content from router caches
also makes it susceptible to replay attacks.

Such replay attacks are problematic in many access control systems. Standard countermeasures
include the use of random, per-message nonces or timestamps. Nonces help ensure
that each message is unique, whereas timestamps protect against interests being
replayed at later points in time. Thus, to mitigate replay attacks, we use both nonces and timestamps.
In particular, each consumer $Cr \in \auth(N)$ must issue an interest with (1) name $N'$,
(2) a randomly generated nonce $r$, and (3) a fresh timestamp $t$. The reason why we use 
both nonces and timestamps is to allow for loosely synchronized clocks and unpredicted network 
latencies. Note that if (1) clocks of consumers, producers, and involved routers in IBAC can be 
perfectly synchronized, and (2) network latencies can be accurately predicted, only timestamps 
are sufficient for replay detection. Moreover, since nonces and
timestamps serve a purpose which is orthogonal to content identification and message
routing, they are included in the interest payload.

Consumer nonces are random $\kappa$-bit values. If a router
receives a duplicate nonce, it can safely assume that the corresponding interest is
replayed and drop it. Let $w$ be a time window associated with authorized content.\footnote{
Determining the proper value of $w$ is outside the scope of 
this paper. However, a logical approach is for routers to use
the lifetime of authorized content as $w$.}
To determine if a duplicate nonce was received, producers (or caches) must maintain
a collection of nonces for each such content. In other words, this 
historical information is necessary to prevent replay attacks. Timestamps themselves 
are not stored, they are only used to determine if the received interest is issued within 
the acceptable time window $w$. Once this time window elapses, all of the stored 
nonces are erased and the content is subsequently flushed from the cache.

Although using nonces and timestamps allows detection of replayed
interests, \Adv\ capturing interests can still use their obfuscated names $N'$ to
fabricate another interest with legitimate $r$ and $t$ values. Therefore, we also
stipulate that $r$ and $t$ should be authenticated via a digital signature; their
signature $\sigma$ is also included in the interest {\tt Payload} field. In order to
bind $r$ and $t$ to their corresponding interest, $N'$ is also included in
the signature computation. $\sigma$ generation and verification 
should be performed using the public and private key pairs associated with 
each access group $\mathbb{G}_i$.

After adding nonces, timestamps, and a signature, interest {\tt Payload} fields take the following form:
\begin{align*}
\mathrm{\tt Payload} = \left( \mathsf{ID}_{\mathbb{G}_i}, r, t, \sigma = \textsf{Sign}_{sk^s_{\mathbb{G}_i}} \left( N' || \mathsf{ID}_{\mathbb{G}_i} || r || t \right) \right)
\end{align*}
where $\mathsf{ID}_{\mathbb{G}_i}$ is the identify of group $\mathbb{G}_i$, and
$sk^s_{\mathbb{G}_i}$ is a signing key distributed to all consumers in $\mathbb{G}_i$.
To verify $\sigma$, the matching public key $pk^s_{\mathbb{G}_i}$ must
be obtained. For the remainder of this paper, we use the term {\em authorization information}
to refer to all information included in interest {\tt Payload} fields for the purpose of supporting IBAC.

One alternative to digital signatures would be to use a keyed hash or a Message Authentication
Code function such as (HMAC) \cite{krawczyk1997hmac}. In this case, consumers and routers would need to
share the key used in the HMAC computation. This means that either consumers or producers 
need to distribute HMAC keys to all involved routers. This, however, is problematic for 
two main reasons: (1) compromising routers leads to HMAC keys leakage, and, more importantly, (2) if consumers 
provide routers with these keys, the former need to know the set of routers that their 
interests traverse before issuing them. Furthermore, since HMAC keys should only be shared among 
involved all entities, i.e., $Cr$ and all routers on $\mathsf{Path}(Cr, P)$, they 
must be distributed securely. Regardless of the 
distribution method used, this incurs extra overhead and complexity compared to simply including, 
in cleartext, signature verification (public) keys in content objects.

Finally, consider the following scenario where two routers $R_1$ and $R_2$ cache content object $CO[N']$ 
which is under IBAC 
protection. Assume that consumer $Cr$ requests $CO[N']$ by sending an interest $I[N']$ with valid 
authorization information that includes $r$ and $t$. Assume that $I[N']$ is satisfied from $R_1$'s cache. 
At the same time, \Adv, 
an eavesdropper between $Cr$ and $R_1$, records $I[N']$. In this case, \Adv\ can replay $I[N']$ 
to $R_2$ and receive $CO[N']$ from the cache since routers do not synchronize stored nonces. 
Therefore, there is no way for $R_2$ to know that $r$ and $t$ were already used at $R_1$. One way of solving 
this problem is to have routers share used nonces lists for each content under IBAC they serve from cache. 
For this method to be effective, such nonces lists need to be securely shared with every single router in the 
network. This might not be feasible in large networks such as the Internet. Another approach 
is to have more accurate synchronized clocks allowing a smaller time window for the aforementioned 
attack to be carried.

\subsection{Authorized Content-Key Binding Rule} \label{sec:acky}
Although the aforementioned method for generating authorization information mitigates replay attacks, 
it also raises several questions. Firstly, how does a router efficiently verify the
signature in interest {\tt Payload} fields? Secondly, and perhaps more importantly, if
a router is able to {\em obtain} the key(s) necessary to verify this signature,
how can the router be sure that such key(s) can be trusted?

To address these problems we propose a mutual trust framework for authorized disclosure. Ghali et al. \cite{ghali2014elements}
first studied the problem of trust in NDN, and ICNs in general, as a means of preventing
content poisoning attacks \cite{ghali2014needle, gasti2012ddos}. Even if routers
can verify content signatures before replying from their cache, it does
not mean that said content is actually \emph{authentic}. Ghali et al. observed
that this verification process requires insight about trust in public keys
(used in verification) that is only known to applications. Consequently,
this requires that all interests must either supply (1) the hash of the public
key used to verify the signature, or (2) the hash of the requested content. In
effect, the interest reflects the trust context of the issuing consumer in a form
enforceable at the network layer. This framework can be viewed as one-way trust
of content by routers. We extend this framework to allow producers to distribute
information about authorized consumers, which can also be enforceable at the
network layer. This allows routers to make trust decisions about individual interests.

Recall that in order for routers to verify which interests are authorized to access
cached content protected under IBAC, the signature in {\tt Payload} must be verified.
To achieve this, producers should include the appropriate verification key
with each IBAC-protected content object. To better understand this, assume the
following scenario. Consumer $Cr \in \mathbb{G}_i$, for $\mathbb{G}_i \subset \auth(N)$,
requests content with name $N$ by issuing an interest with obfuscated name $N'$,
and $\mathsf{ID}_{\mathbb{G}_i}$, $r$, $t$ and $\sigma$ in {\tt Payload} as
described in Section \ref{sec:replay}. Assume that the matching content is not cached
anywhere in the network. Once this interest reaches the producer $P$, the latter verifies 
$\sigma$ and replies with the content that also includes the verifying
key $pk^s_{\mathbb{G}_i}$.\footnote{The content object signature must also be computed over $pk^s_{\mathbb{G}_i}$ to bind it to the message.} Router $R$ will then cache $pk^s_{\mathbb{G}_i}$ along with the content itself.
Once another interest for $N'$ is received, $R$ uses the cached
$pk^s_{\mathbb{G}_i}$ to verify $\sigma$ and returns the corresponding cached content object. 

We formalize this in the following policy, denoted as the Authorized Content-Key Binding (ACKB) rule:

\vspace{0.2cm}
\centerline{\fbox{\parbox{0.968\columnwidth}{\textcolor{blue}{\sf {\bf ACKB:} Cached content protected under IBAC must reflect the verification key associated with the authorization policy.}}}}
\vspace{0.2cm}

\begin{algorithm}[t]
\caption{{\sf InterestGeneration}}\label{alg:interest-generation}
\begin{algorithmic}[1]
\scriptsize
\STATE \textbf{INPUT:} $\mathsf{routable\_prefix}$, $N$, $k_{\mathbb{G}_i}$, $pk^s_{\mathbb{G}_i}$, $sk^s_{\mathbb{G}_i}$, $\kappa$

\STATE $\mathsf{ID}_{\mathbb{G}_i} \gets \mathsf{H}(k_{\mathbb{G}_i})$
\STATE $N' \gets /\mathsf{routable\_prefix}/f(k_{\mathbb{G}_i}, \textsf{Suffix}(N, \mathsf{routable\_prefix}))$
\STATE $r \xleftarrow{\$} \{0,1\}^{\kappa}$
\STATE $t \gets \mathsf{CurrentTime}()$
\STATE $\sigma \gets \mathsf{Sign}_{sk^s_{\mathbb{G}_i}} \left( N' || \mathsf{ID}_{\mathbb{G}_i} || r || t \right)$
\STATE $\mathrm{\tt Payload} := \left( \mathsf{ID}_{\mathbb{G}_i}, r, t, \sigma \right)$
\STATE \textbf{return} $I[N'] := \left( N', \mathrm{\tt Payload} \right)$
\end{algorithmic}
\end{algorithm}

\begin{algorithm}[t]
\caption{{\sf ContentObjectGeneration}}\label{alg:content-generation}
\begin{algorithmic}[1]
\scriptsize
\STATE \textbf{INPUT:} $I[N'] := \left( \mathsf{routable\_prefix}, N', \mathrm{\tt Payload} \right)$

\STATE $\left( \mathsf{ID}_{\mathbb{G}_i}, r, t, \sigma \right) := \mathrm{\tt Payload}$
\STATE $pk^s_{\mathbb{G}_i} \gets \mathsf{LoopupVerificationKeyForID}(\mathsf{ID}_{\mathbb{G}_i})$
\IF{$\mathsf{Verify}_{pk^s_{\mathbb{G}_i}}(\sigma)$}
	\STATE $k^e_{\mathbb{G}_i} \gets \mathsf{LookupDecryptionKeyForID}(\mathsf{ID}_{\mathbb{G}_i})$
	\STATE $N \gets \mathsf{Dec}(k^e_{\mathbb{G}_i}, \mathsf{Suffix}(N', \mathsf{routable\_prefix}))$
	\STATE $\mathsf{data} \gets \mathsf{RetrieveContent}(N)$
	\STATE \textbf{return} $CO[N'] := (N', \mathsf{data}, pk^s_{\mathbb{G}_i})$
\ELSE
	\STATE Drop $I[N']$
\ENDIF
\end{algorithmic}
\end{algorithm}

\begin{algorithm}[t]
\caption{{\sf RouterAuthorizationCheck}}\label{alg:router-check}
\begin{algorithmic}[1]
\scriptsize
\STATE {\textbf{INPUT:} $I[N']$, cached $CO[N']$, $\mathbb{B}$}

\STATE $\left( \mathsf{ID}_{\mathbb{G}_i}, r, t, \sigma \right) := \mathrm{\tt Payload}$
\STATE $(N', \cdot, pk^s_{\mathbb{G}_i}) := CO[N']$
\IF{$\mathbb{B}[N']$ contains $r$}
	\STATE Drop $I[N']$; \textbf{return} {\sf Fail}
\ELSE
	\IF{Timestamp $t$ is invalid}
		\STATE Drop $I[N']$; \textbf{return} {\sf Fail}
	\ELSE
		\IF{$\mathsf{Verify}_{pk^s_{\mathbb{G}_i}}(\sigma)$}
			\STATE $\mathbb{B}[N'] := \mathbb{B}[N'] \cup r$
			\STATE \textbf{return} {\sf Pass}
		\ELSE
			\STATE Drop $I[N']$; \textbf{return} {\sf Fail}
		\ENDIF
	\ENDIF
\ENDIF
\end{algorithmic}
\end{algorithm}

\begin{figure}[t]
  \centering
  \includegraphics[width=\columnwidth]{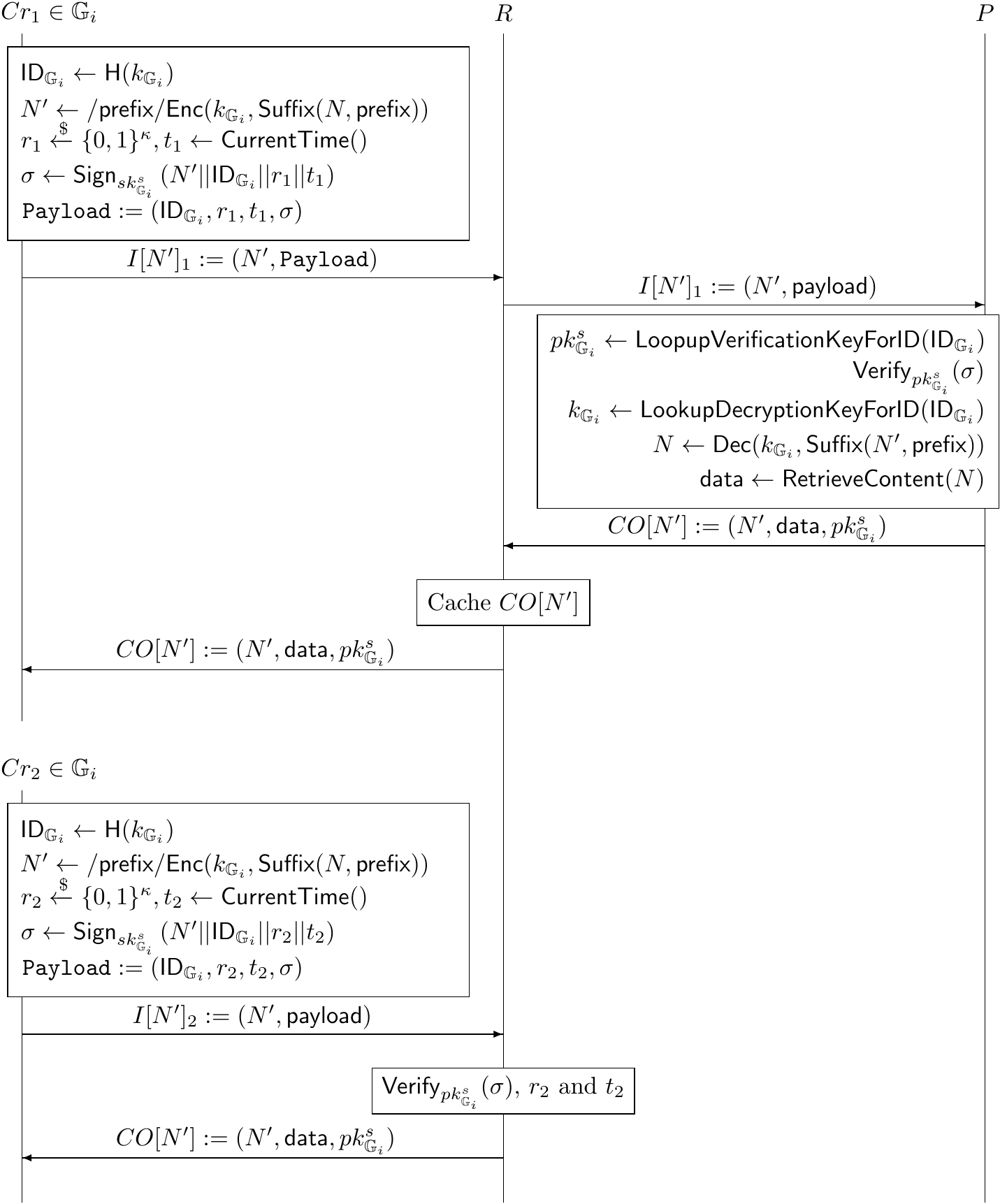}
  \caption{Consumer and producer exchanges for IBAC-protected content.}
  \label{fig:sequence}
\end{figure}

The protocol for IBAC-protected content retrieval relies on this rule. Algorithms
\ref{alg:interest-generation} and \ref{alg:content-generation} outline the interest and
content object generation procedures. Note that the function $\mathsf{Suffix}(N, \mathsf{routable\_prefix})$
returns all name components of $N$ except the ones included in $\mathsf{routable\_prefix}$.\footnote{For instance, $\mathsf{Suffix}($\ccnname{/edu/uci/ics/home.html}, \ccnname{/edu/uci/}$)$ would return \ccnname{ics/home.html}.}
Also, the router verification procedure is outlined
in Algorithm \ref{alg:router-check}. If this procedure returns {\tt Pass}, then the content object
found in the cache is forwarded downstream to the associated interface. Note
that Algorithms \ref{alg:interest-generation}, \ref{alg:content-generation}, and \ref{alg:router-check}
use obfuscation key $k_{\mathbb{G}_i}$ and signing key pairs $(pk^s_{\mathbb{G}_i}, sk^s_{\mathbb{G}_i})$. For 
completeness, a complete sequence diagram showing multiple
interest-content exchanges is shown in Figure \ref{fig:sequence}. Both consumers belong to
the same access group, i.e., $Cr_1, Cr_2 \in \mathbb{G}_i$.

In Appendix \ref{app:proofs}, we show that this mutual trust framework for authorized
disclosure enables IBAC with stronger security guarantees in the presence of replay attacks.

\subsection{Serving Content to Multiple Access Groups} \label{sec:multigroup}
One problem with encryption-based name obfuscation occurs when a content object with
name $N$ is accessible by different groups. According to Algorithms \ref{alg:interest-generation}
and \ref{alg:content-generation}, the obfuscated name $N'$ contains a suffix encrypted with keys
associated with each access group. Therefore, a single content object might have several names
depending on the number of groups authorized to access it. Since routers employs exact
matching for cache lookup\footnote{In CCN, not in NDN.}, several copies of the same content could
possibly be cached.

To solve this problem, content objects should have the exact same name regardless of 
access control groups permitted access. This can be achieved using the hash-based name obfuscation function
described in Section~\ref{sec:hash-obfuscation}. However, cached content needs to 
contain \emph{every} authorization signature verification key that could be used to access said content.
In other words, producers need to provide all possible public keys that can be used to
access the content under IBAC protection. Consider the following scenario: a content
object $CO[N]$ is accessible by two access groups $\mathbb{G}_i$ and $\mathbb{G}_j$.
In this case, the producer needs to provide both $pk^s_{\mathbb{G}_i}$ and $pk^s_{\mathbb{G}_j}$ with $CO[N']$, i.e.,
\begin{align*}
CO[N'] := (N', \mathsf{data}, pk^s_{\mathbb{G}_i}, pk^s_{\mathbb{G}_j})
\end{align*}

Whenever a router $R$ caching $CO[N']$ receives an interest issued by a consumer in
any of the authorized access groups, $R$ uses the group identity included in the
{\tt Payload} field to determine $\sigma$'s verification key.

Note that content object sizes might increase significantly depending on how many 
groups are allowed access. We do not discuss this issue further, since the trade-off 
between having multiple cached versions of the same content and having longer content 
objects carrying all verification keys is ultimately the application's decision.

\subsection{IBAC Variations} \label{sec:variations}
We do not claim that any of the IBAC variations discussed above is superior to another. Each
has its own strengths and weaknesses. However, to help guide the decision about which variation
to use, we make the following claims based on the application needs and assumptions. Note that 
some claims provide privacy as well as access control.
\begin{enumerate}
	\item If {\em replay attacks} are not a concern, then consumers only need to use a name obfuscation function and include their group identity in the {\tt Payload}.
	\item If {\em replay attacks} are plausible and {\em name privacy} is a concern, then name obfuscation must be used and authorization information, as described in Section \ref{sec:replay}, must be included in interest {\tt Payload} fields.
	\item\label{clm:no-privacy} If {\em replay attacks} are plausible but {\em name privacy} is not a concern, then only authorization information is sufficient.
\end{enumerate}
Claim \ref{clm:no-privacy} might seem counterintuitive with the idea of IBAC. Recall, however, that router
authorization checks prevent unauthorized consumers from retrieving cached content under
IBAC protection. Even if content name is not obfuscated, \Adv\ cannot forge {\tt Payload} authorization
information, and therefore cannot violate IBAC protection guarantees.

\subsection{Revocation}
Generally speaking, revocation is a challenge in all access control schemes involving secrets shared
among group members. Recall that all consumers belonging to the same access control group in IBAC
share the same obfuscation keys. If one of them leaves the group\footnote{For instance, consumers
not renewing their subscription for a certain service.}, the producer will have to create a new
key and distribute it to all remaining authorized consumers. We will not discuss this
issue further since we believe it is not part of the core access control \emph{protocol}.

Moreover, in-network caching can cause IBAC content to be accessed by revoked consumers.
Assume content $CO[N]$ is under access control and has a cached version in router $R$. Assume
consumer $Cr$, connected (directly or indirectly) to $R$, is authorized to access $CO[N]$.
However, while $CO[N]$ is cached, $Cr$'s access is revoked. At the same time, the latter sends an
interest requesting $CO[N]$. In this case, $R$ will grant access and reply with $CO[N]$ from its
cache. This is due to the fact that the cached version of $CO[N]$ is not updated with the
correct authorization information (i.e., verification key(s)). However, this can be solved by
setting the \texttt{ExpiryTime} field of $CO[N]$ to a value that reflects consumer revocation frequency.

Online revocation protocols, such as OCSP \cite{myers1999online}, would induce extra
communication between $R$ and $P$, which nearly defeats the purpose of the cache entirely. In this
case, $R$ would be better suited forwarding the interest upstream to $P$. Another option for the
producer would be to distribute certificate revocation lists (CRLs) \cite{cooper2008internet}
with every fresh content. This, however, introduces further issues for routers and consumers.
Firstly, routers would need to store CRLs and keep them updated frequently. Secondly, authorized consumers would need their own public and private key pair to compute $\sigma$. Finally, routers would need to perform
additional verifications against the CLR. Overall, this approach suffers from increased storage,
consumer management, computation, and bandwidth complexity.

\begin{table*}[t]
\small
\begin{center}
\caption{Overview of per-interest IBAC-induced computational overhead for routers and producers.}
\label{tab:overhead}
\begin{tabular}{|c|c|p{5cm}|p{5cm}|} \hline
\multicolumn{2}{|c|}{\multirow{2}{*}{\textbf{IBAC Variation}}} & \multicolumn{2}{c|}{\textbf{IBAC-induced Computation Overhead}} \\ \cline{3-4}
\multicolumn{2}{|c|}{} & \multicolumn{1}{c|}{\textbf{Routers}} & \multicolumn{1}{c|}{\textbf{Producers}} \\ \hline
\multirow{2}{*}{\textbf{Name Obfuscation}} & \textbf{Encryption} & None & One decryption \\ \cline{2-4}
                                  & \textbf{Hash}       & None & One hash table lookup \\ \hline
\multirow{2}{*}{\textbf{Interest Signatures}} & \textbf{Encryption} & One signature verification, one nonce and timestamp verification & One decryption, one signature verification, Two hash table lookups (decryption key and signing key resolution) \\ \cline{2-4}
                                     & \textbf{Hash}       & One signature verification, one nonce and timestamp verification, one hash table lookup (signing key resolution) & One signature verification, three hash table lookups (decryption key, signing key and name resolution) \\ \hline
\end{tabular}
\end{center}
\end{table*}

\section{Analysis and Evaluation} \label{sec:analysis}


In this section, we analyze the overhead induced by each variation of the proposed IBAC scheme.

\subsection{Computational Overhead}
We first focus on the computational overhead for routers and producers. This overhead is captured in
terms of cryptographic and data structure operations, e.g., signature verification and hash table
lookup costs. Table \ref{tab:overhead} summarizes these results. To further understand the computational
overhead, we compare two cases: (1) when routers perform authorization checks, and (2) when they do not.
Let $\tau_{overhead} = \tau_{check} + \tau_{verify} + \tau_{update}$ be the overhead induced
by the authorization check when routers receive interests, where $\tau_{check}$ is the time required to
check for nonce duplication and timestamp staleness, $\tau_{verify}$ is the time to verify the {\tt Payload}
signature, and $T_{update}$ is the time to update the nonce data collection. Since cache lookup and
interest forwarding are performed regardless of whether or not routers perform authorization checks, we omit
them from this equation. Similarly, $\tau_{check}$ and $\tau_{update}$ are negligible when compared
to the cost of signature verification $\tau_{verify}$; thus, they are also excluded.

A router incurs a computational cost of $\tau_{overhead}$ for every received interest requesting
content under IBAC protection. Therefore, we quantify $\tau_{overhead}$ by measuring the time it
takes to perform a single signature verification. We also experiment with batch verification
techniques to better amortize the cost of signature verification across series of interests.
While this naturally increases content retrieval latency since signatures are accumulated in
case of batch verification, it reduces router computational overhead.  
Table \ref{tab:verificationtimes} shows the amount of improvement
using a variety of signature verification algorithms.
Note that, when modeling interest arrival rates using a Poisson distribution, both 
individual and batch signature verification incur nearly the same overhead in certain conditions, 
as we will show below.

\begin{table}[t]
\small
\begin{center}
\caption{Individual and batch ElGamal signature verification times.}
\label{tab:verificationtimes}
\begin{tabular}{|c|c|c|c|c|c|} \hline
{\bf Key} & {\bf Batch} & {\bf Sig.} & {\bf Indiv.} & {\bf Batch}  & \multirow{2}{*}{{\bf Improved}} \\
{\bf Size} & {\bf Size} & {\bf Size} & {\bf Time} & {\bf Time} & \\ \hline
\hline
1024b & 10 & 512KB & 0.599s & 0.322s & 46\% \\ \hline
1024b & 10 & 8MB & 0.888s & 0.615s & 30\% \\ \hline
1024b & 50 & 512KB & 2.918s & 1.579s & 46\% \\ \hline
1024b & 50 & 8MB & 4.315s & 2.991s & 30\% \\ \hline
\hline
2048b & 10 & 512KB & 4.065s & 2.207s & 46\% \\ \hline
2048b & 10 & 8MB & 4.104s & 2.269s & 45\% \\ \hline
2048b & 50 & 512KB & 20.081s & 11.029s & 45\% \\ \hline
2048b & 50 & 8MB & 21.301s & 12.536s & 41\% \\ \hline
\hline
3072b & 10 & 512KB & 12.406s & 6.789s & 45\% \\ \hline
3072b & 10 & 8MB & 12.804s & 7.122s & 44\% \\ \hline
3072b & 50 & 512KB & 60.174s & 32.877s & 45\% \\ \hline
3072b & 50 & 8MB & 64.347s & 35.601s & 45\% \\ \hline

\end{tabular}
\end{center}
\end{table}

Denial of service (DoS) is an obvious concern if routers perform authorization checks.
Let $\lambda$ be the rate of arrival interests for IBAC-protected content cached in router
$R$, and let $\mu$ be the service rate for interests, i.e., the rate at which interests
are processed (parsed, verified, etc.). If $\mu < \lambda$, then the router will be over
encumbered with interests to process \cite{gross2008fundamentals}. We envision that in
legitimate scenarios without malicious entities generating interests with fake 
authorization information, only a small
percentage $\delta$ of arrival interests will be requesting content under IBAC protection.
To assess how susceptible routers are to DoS attacks induced by IBAC authorization checks,
we empirically analyze the effect of $\delta$ on the interest service rate of a router.
These service rates, which use different signature verification techniques -- individual and batch --
denoted $\mu_S$ and $\mu_B$, respectively, are shown in Figure \ref{fig:service-rates}.

\begin{figure}
\center
\includegraphics[width=\columnwidth]{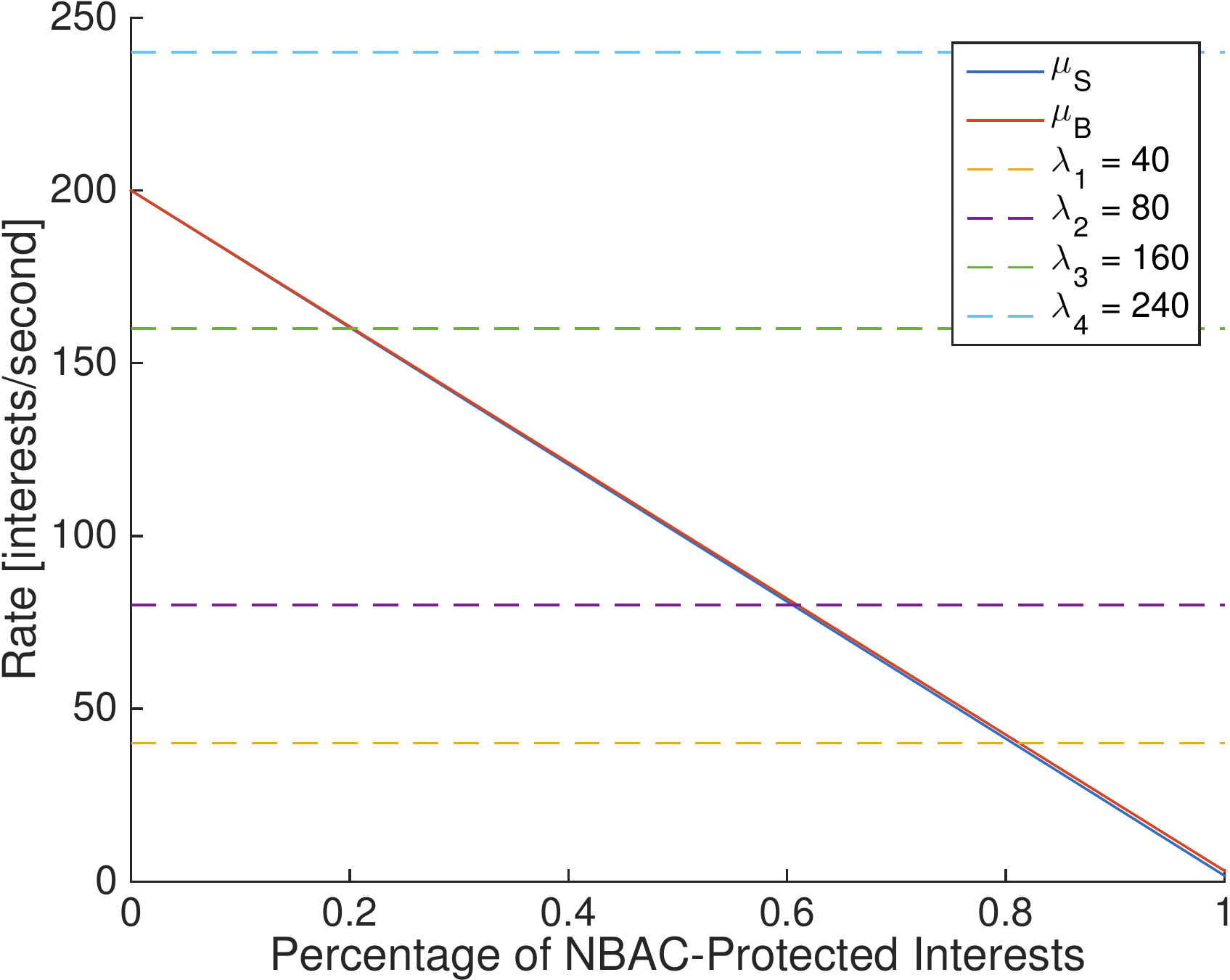}
\caption{Interest service rates for various percentages of IBAC-protected interests.}
\label{fig:service-rates}
\end{figure}

We assume that interests arrive at a base rate of $\lambda_1 = 40$ \cite{carofiglio2011modeling}; larger
values for $\lambda$ are provided to see at which point $\mu < \lambda$ due to authorization checks. By the
exponential property of the Poisson process, $\mu$ is calculated as follows:
\begin{align*}
\mu = \frac{1 - \delta}{\tau_{process}} + \frac{\delta}{\tau_{process} + \tau_{verify}},
\end{align*}
where $\tau_{process}$ represents interest processing time not including signature
verification\footnote{$\tau_{process} = 1/mu$ for interests not requesting IBAC-protected content.},
and $\tau_{verify}$ is the time required to perform individual or batch signature verification.
In our experiment, we assume a constant $\tau_{process} = 0.005$s and only vary
$\tau_{verify}$. To do so, we assume a key size of $1024$b, batch size of 10, and signature size of $512$KB.
According to Table \ref{tab:verificationtimes}, this results in $\tau_{verify} = 0.599$s
and $\tau_{verify} = 0.322$s for individual and batch verification, respectively.
Our experiments show that the decay of $\mu$ as a function of $\delta$
is almost identical for both batch and verification techniques. This is due to the
fact that only a small fraction of interests are affected by the verification step.
Furthermore, our results show that $\mu > \lambda$ is true, i.e., the router servicing
process is stable for reasonable interest arrival rates. Our experiments show that
$\mu < \lambda$ when $\lambda = 160$ and $\delta \geq 0.2$. Moreover, when a Poisson process 
is assumed, both individual and batch signature verification perform similarly for small values 
of $\delta$. However, batch signature verification prove to be advantageous in larger $\delta$ 
values. For instance, for $\delta = 0.2$, using batch verification provides less than $1\%$ service 
rate improvement, where it provides $3\%$ and $46\%$ for $\delta$ values equal to $0.8$ and $1$, 
respectively.

\subsection{Storage Overhead}
Storage overhead varies from producer to router. If hash-based name obfuscation is used,
producers incur the cost of maintaining a hash table to map obfuscated names to their
original values. However, if content name contains variable name components, e.g., query
string-like values in URIs, the hash table size can grow significantly since it has to
contain all possible variations. Moreover, producers must bear the storage cost of IBAC
access group keys if encryption-based obfuscation functions are used. Similarly, routers
must bear the cost of storing variable-length tuples of key identities $\mathsf{ID}_{\mathbb{G}_i}$
and the actual verification keys $pk^s_{\mathbb{G}_i}$, along with a theoretically unbounded
collection of nonces for each IBAC-protected content. This finite amount of storage
can be abused to mount DoS attacks on routers.

\subsection{Bandwidth Overhead}
In terms of bandwidth overhead, each interest and content object is expanded to include
additional authorization information, e.g., interest payloads with authorization information and content
objects with authorization keys. Interests without authorization payloads will only increase
(or decrease) by the expansion factor of the obfuscated name. If authorization payloads are
included, then interest messages will grow by $|r| + |t| + |\sigma| + |\mathsf{ID}_{\mathbb{G}}|$,
where $|r| = \kappa$. Content object $CO[N]$ grows with length
$\sum_{i=1}^L |pk^s_{\mathbb{G}_i}|$, where $L$ is the number of access groups allowed to
access $CO[N]$ and $|pk^s_{\mathbb{G}_i}|$ is the public key size associated with group $\mathbb{G}_i$.


\section{Conclusion} \label{sec:conclusion}
We studied the problem of access control in ICNs. We proposed an Interest-Based Access
Control (IBAC) scheme that supports hash- and encryption-based name obfuscation. We addressed
the problem of replay attacks by formulating a mutual trust framework between
producers and consumers -- enforced in the network-layer -- that enables routers to perform
authorization checks before satisfying interests from cache. We assessed the
computational, storage, and bandwidth overhead induced by each variant of the proposed
IBAC scheme. Ultimately, we believe that our work brings ICNs one step closer to
fulfilling their promise of a more secure networking paradigm.

\balance

\bibliographystyle{abbrv}
\bibliography{references}

\appendix
\section{Proofs of Security} \label{app:proofs}
In this section, we prove the security properties of IBAC with and without authorized
disclosure with respect to the adversarial model described in Section \ref{sec:secmodel}.
In the following, let $N$ be the name of a content object under IBAC protection and generated by $P$. 
Also, let adversary $\mAdv = (\mathcal{P}_{\mAdv} \setminus \{P\}, \mathcal{C}_{\mAdv} \setminus \auth(N), \mathcal{R}_{\mAdv} \setminus \mathsf{Path}(Cr, P))$, where $Cr \in \auth(N)$. 
\begin{theorem}\label{thm:proof1}
The IBAC scheme without authorized disclosure is secure, {\em but subject to replay attacks}, 
against $\mAdv$ if an indistinguishably\hyp{}secure (IND-secure) deterministic 
encryption algorithm is used for name obfuscation.
\end{theorem}

\noindent
{\bf Note}: IND 
security is typically identical to CPA security in the public-key setting 
since the adversary is assumed to have access to the public key \cite{katz2014introduction}. 
In this case, neither the encryption nor decryption key is known to $\mAdv$.

\begin{proof}
Let $\Pi = (\mathsf{Gen}, \mathsf{Enc}, \mathsf{Dec})$ be an IND-secure (deterministic) encryption scheme consisting of three probabilistic polynomial time algorithms $\mathsf{Gen}$, $\mathsf{Enc}$, and $\mathsf{Dec}$ for key generation, encryption, and decryption, respectively. Let $k_e$ and $k_d$ be the encryption and decryption keys produced by $\mathsf{Gen}$. For any interest name $N$, it holds that $\mathsf{Dec}(k_d, \mathsf{Enc}(k_e, N)) = N$. Let \Adv\ be any probabilistic polynomial adversary. The definition of the eavesdropping indistinguishability experiment, adapted for plaintext interest messages, denoted $\mathsf{Exp}_{\mAdv, \Pi}^{\mathsf{ind}}$, is as follows:
\begin{enumerate}
	\item \Adv\ is given input $1^\kappa$ and outputs a pair of interest names $N_0$
	and $N_1$, and $k_e$ and $k_d$ are computed by running $\mathsf{Gen}(1^\kappa)$.
	\item A single bit $b \gets \{0,1\}$ is chosen uniformly at random. The challenger
	computes the ciphertext $c \gets \mathsf{Enc}(k_e, N_b)$, which is given to \Adv.
	\item \Adv\ outputs a single bit $b'$.
	\item The output of the experiment is said to be $1$ if $b' = b$ and $0$ otherwise.
\end{enumerate}
Let $\mathsf{Exp}_{\mAdv, \Pi}^{\mathsf{ind}}(\kappa, b)$ be the same experiment run
but where bit $b$ is given as an input value. By the definition of IND-security, it follows that
\begin{align*}
|\Pr[\mathsf{Exp}_{\mAdv, \Pi}^{\mathsf{ind}}(\kappa, 1) = 1] - \Pr[\mathsf{Exp}_{\mAdv, \Pi}^{\mathsf{ind}}(\kappa, 0) = 1]| \leq \epsilon(\kappa),
\end{align*}
for some negligible function $\epsilon$. Recall that $\mathsf{Guess}$ is the event that
\Adv\ correctly guesses the obfuscated version a content name. The probability of
\Adv\ decrypting a message is at least $\Pr[\mathsf{Guess}]$. Therefore, 
the event when \Adv\ successfully guesses the obfuscated version of the name,
is when \Adv\ outputs $b' = 1$ when $b = 1$ and $b' = 0$ when $b = 0$. Thus,
\begin{align*}
\Pr[\mathsf{Guess}] &= |\Pr[\mathsf{Exp}_{\mAdv, \Pi}^{\mathsf{ind}}(\kappa, 1) = 1] \\
&\quad- \Pr[\mathsf{Exp}_{\mAdv, \Pi}^{\mathsf{ind}}(\kappa, 0) = 1]| \\
&\leq \epsilon(\kappa)
\end{align*}
This concludes the proof.
\end{proof}

\begin{theorem}
The IBAC scheme with authorized disclosure is secure, {in presence of replay attacks},
against $\mAdv$ if an indistinguishably\hyp{}secure (IND-secure) deterministic
encryption algorithm is used with an existentially unforgeable signature scheme.
\end{theorem}
\begin{proof}
In Theorem \ref{thm:proof1}, we proved that $\Pr[\mathsf{Guess}] \leq \epsilon(\kappa)$. It
is easy to see that the additional {\tt Payload} information -- the random nonce, timestamp,
and signature -- are all distinct for each interest. Therefore, including this information
leaks no information that improves the adversaries advantage or improves $\Pr[\mathsf{Guess}]$.

We now assess $\Pr[\mathsf{Bypass}]$. Recall that this event occurs when \Adv\ bypasses the
authorization check at a router to recover content from a cache. Without knowledge of $sk^s_{\mathbb{G}_i}$,
this only occurs if \Adv\ is able to forge the {\tt Payload} signature. By definition of the
existentially unforgeable signature scheme, \Adv\ is not able to generate an input set
$(\hat{N'}, \hat{\mathsf{ID}_{\mathbb{G}_i}}, \hat{r}, \hat{t}) \neq (N', \mathsf{ID}_{\mathbb{G}_i}, r, t)$
such that $\mathsf{Verify}_{pk^s_{\mathbb{G}_i}}(\hat{\sigma})$ occurs with non-negligible probability.
Thus, $\Pr[\mathsf{Bypass}] \leq \epsilon(\kappa)$. Finally, since the sum of two negligible probabilities is
also negligible, then $\Pr[\mathsf{Guess} + \mathsf{Bypass}] \leq \epsilon(\kappa)$ .
\end{proof}

\end{document}